%% file: kkl.tex
\documentclass[final, 11pt]{article}

\usepackage{amsmath,amssymb,amsfonts}
\usepackage{mathtools}
\usepackage{epsfig}
\usepackage{latexsym,nicefrac,bbm}
\usepackage{xspace}
\usepackage{color,fancybox,graphicx,url,subfigure}
\usepackage{hyperref,enumitem}

\def\showauthornotes{0}

\def\showdraftbox{0}
\input{macros}

\newcommand{\Mnote}{\Authornote{M}}

\newcommand{\Snote}{\Authornote{S}}

\newcommand{\I}{\mathbb{I}}
\newcommand{\cart}{\Box}
\newcommand\Ent{\mathop{\mbox{\bf Ent}}}
\newcommand\vol{\mathop{\mbox{Vol}}}
\newcommand\Gk{G^{\cart k}}
\newcommand\opt{\mathop{\mbox{Opt}}}
\newcommand{\varj}[1]{{\mbox{\bf Var}_{#1}}}
\newcommand{\e}{\epsilon}
\newcommand{\SC}{\textsf{Sparsest Cut}}

\newcommand{\defeq}{\stackrel{\textup{def}}{=}} 
\newcommand{\eps}{\varepsilon}
\renewcommand{\epsilon}{\varepsilon}
\newcommand{\nfrac}{\nicefrac}
\renewcommand{\R}{\rea}

\title{Cuts in Cartesian Products of Graphs}
\author{ Sushant Sachdeva\thanks{Research Fellow, Simons Institute for
    the Theory of Computing. UC Berkeley, USA. Part of this work was
    done when this author was a graduate student at the Department of
    Computer Science, Princeton University. Email: {\tt
      sachdeva@eecs.berkeley.edu}}
\and 
Madhur Tulsiani \thanks{
Toyota Technological Institute at Chicago. \texttt{madhurt@ttic.edu}}
}

\begin{document}
\maketitle

\draftbox

\begin{abstract}
The $k$-fold Cartesian product of a graph $G$ is defined as a graph on
$k$-tuples of vertices, where two tuples are connected if they form an
edge in one of the positions and are equal in the rest. Starting with
$G$ as a single edge gives $\Gk$ as a $k$-dimensional hypercube. We
study the distributions of edges crossed by a cut in $\Gk$ across the
copies of $G$ in different positions. This is a generalization of the
notion of \emph{influences} for cuts on the hypercube.

\smallskip We show the analogues of results of Kahn, Kalai, and Linial
(KKL Theorem \cite{KahnKL88}) and that of Friedgut (Friedgut's Junta
theorem \cite{Friedgut98}), for the setting of Cartesian products of
arbitrary graphs. Our proofs extend the arguments of
Rossignol~\cite{Rossignol06} and of Falik and
Samorodnitsky~\cite{FalikS07}, to the case of arbitrary Cartesian
products.
We also extend the work on studying isoperimetric constants for these
graphs \cite{HoudreT96, ChungT98} to the value of semidefinite
relaxations for edge-expansion.  We connect the optimal values of the
relaxations for computing expansion, given by various semidefinite
hierarchies, for $G$ and $\Gk$.
\end{abstract}

\Snote{To do:
\begin{enumerate}
\item Should we give the def for regular graphs, and then say
  reversible markov chains in the abstract?
\end{enumerate}
}


\input{intro}
\input{prelims}
\input{theorems}

\input{tightness}
\input{irregular}

\section*{Acknowledgements}
We thank Ryan O'Donnell and Rishi Saket for helpful
comments. We are also grateful to Elchanan Mossel and Ryan
O'Donnell for several relevant pointers to the literature.

\bibliographystyle{plain}
\bibliography{kkl}
\appendix

\input{applications}

\end{document}

%% file: macros.tex
\newtheorem{theorem}{Theorem}[section]

\newtheorem{definition}[theorem]{Definition}
\newtheorem{lemma}[theorem]{Lemma}

\newtheorem{proposition}[theorem]{Proposition}
\newtheorem{corollary}[theorem]{Corollary}
\newtheorem{claim}[theorem]{Claim}

\def\FullBox{\hbox{\vrule width 6pt height 6pt depth 0pt}}

\def\qed{\ifmmode\qquad\FullBox\else{\unskip\nobreak\hfil
\penalty50\hskip1em\null\nobreak\hfil\FullBox
\parfillskip=0pt\finalhyphendemerits=0\endgraf}\fi}

\def\qedsketch{\ifmmode\Box\else{\unskip\nobreak\hfil
\penalty50\hskip1em\null\nobreak\hfil$\Box$
\parfillskip=0pt\finalhyphendemerits=0\endgraf}\fi}

\newenvironment{proof}{\begin{trivlist} \item {\bf Proof:~~}}
   {\qed\end{trivlist}}

\newcommand\rea{\mathbb R}

\newcommand\N{\mathbb N}
\newcommand\R{\mathbb R}

\newcommand\B{\{0,1\}}      

\newcommand{\V}[1]{\mathbf{#1}}

\newcommand{\marginlabel}[1]%
{\mbox{}\marginpar{\it{\raggedleft\hspace{0pt}#1}}}

\newcommand{\pair}[1]{\left\langle{#1}\right\rangle} 

\definecolor{Mygray}{gray}{0.8}

 \ifcsname ifcommentflag\endcsname\else
  \expandafter\let\csname ifcommentflag\expandafter\endcsname
                  \csname iffalse\endcsname
\fi

\ifnum\showauthornotes=1

\else

\fi

\ifnum\showauthornotes=1
\newcommand{\Authornote}[2]{{\sf\small\color{red}{[#1: #2]}}}
\newcommand{\Authoredit}[2]{{\sf\small\color{red}{[#1]}\color{blue}{#2}}}
\newcommand{\Authorcomment}[2]{{\sf \small\color{gray}{[#1: #2]}}}
\newcommand{\Authorfnote}[2]{\footnote{\color{red}{#1: #2}}}
\newcommand{\Authorfixme}[1]{\Authornote{#1}{\textbf{??}}}
\newcommand{\Authormarginmark}[1]{\marginpar{\textcolor{red}{\fbox{
#1:!}}}}
\else
\newcommand{\Authornote}[2]{}
\newcommand{\Authoredit}[2]{}
\newcommand{\Authorcomment}[2]{}
\newcommand{\Authorfnote}[2]{}
\newcommand{\Authorfixme}[1]{}
\newcommand{\Authormarginmark}[1]{}
\fi

\newcommand\calD{\mathcal{D}}

\newcommand{\inparen}[1]{\left(#1\right)}             
\newcommand{\inbraces}[1]{\left\{#1\right\}}           
\newcommand{\insquare}[1]{\left[#1\right]}             

\newcommand{\suchthat}{\;\middle\vert\;}

\newcommand{\from}{:}

\newcommand\pr{\mathop{\mbox{\bf Pr}}}
\newcommand\av{\mathop{\mbox{\bf E}}}
\newcommand\var
{
  \mathop{
    \mathchoice
    {\mbox{\bf Var}}
    {{\bf Var}}
    {\mbox{\bf \scriptsize Var}}
    {}
  }
}

\newcommand{\Ex}[2]{\av_{{#1}}\left[{#2}\right]}

\def\abs#1{\left| #1 \right|}
\newcommand{\norm}[1]{\ensuremath{\left\lVert #1 \right\rVert}}

\newlength{\pgmtab}  
 {
	\begin{enumerate}}{\end{enumerate}}

\newcounter{lecnum}

\newlength{\tpush}
\ifnum\showdraftbox=1
\newcommand{\draftbox}{\begin{center}
  \fbox{%
    \begin{minipage}{2in}%
      \begin{center}%
          \large\textsc{Working Draft}\\%
        Please do not distribute%
      \end{center}%
    \end{minipage}%
  }%
\end{center}
\vspace{0.2cm}}
\else
\newcommand{\draftbox}{}
\fi

%% file: intro.tex
\section{Introduction}
The Cartesian product of two graphs $G$ and $H$ is defined as the
graph $G \cart H $ on the vertex set $V(G) \times V(H)$, with the
tuples $\{ (i_1, i_2), (j_1,j_2)\}$ forming an edge if $\{i_1,j_1\}$
forms an edge in $G$ and $i_2 = j_2$, or $\{i_2,j_2\}$ forms an edge
in $H$ and $i_1 = j_1$.  The $k^\textrm{th}$ power of a graph according to
this product is defined by associativity as $G^{\cart k} = G \cart
G^{\cart (k-1)}$. The notion is a well known and well-studied one in
graph theory (see \cite{Vizing63}) for example.
\Snote{The cartesian product is for regular graphs. Should we directly
define $\Gk$ instead?}

Certain special cases of this product are particularly interesting to
consider. For example, when $G$ is just an edge, $G^{\cart k}$ is the
$k$-dimensional hypercube on $2^k$ vertices, which is perhaps the best
known example. Starting with $G$ as the $n$-cycle, $G^{\cart k}$ gives
the $k$-dimensional torus. Similarly, starting from a path yields a
grid.

Isoperimetric questions for Cartesian products of graphs have been
studied by various authors.  Houdr\'{e} and Tetali \cite{HoudreT96}
compute various isoperimetric invariants for the analogue of this
notion in the case of Markov chains, where the product of two Markov
chains is defined by a process which randomly selects one of the two
chains and make a transition according to it. \Snote{Define the
  $k$-fold product.}  Chung and Tetali
\cite{ChungT98} considered the combinatorial version and gave a
combinatorial proof of the bound on the conductance of Cartesian
products in terms of those of the starting graph (see references in
\cite{ChungT98} for previous work on special cases).

We study the question of extending some of the isoperimetric
inequalities known for the case of the hypercube in terms of the
\emph{influences} of variables, to the case of Cartesian products of
graphs. For a Boolean function on the $k$-dimensional hypercube, the
influence of the function in the $i^{\textrm{th}}$ coordinate is
defined as the probability over a random input that changing the
$i^{\textrm{th}}$ bit changes the value of the function. The total
influence of the function is the sum of the influences along all the
$k$ coordinates. Viewing a Boolean function as cut on the hypercube,
the influence along the $i^{\textrm{th}}$ coordinate is simply the
fraction of the edges along the the $i^{\textrm{th}}$ direction (which
correspond to changing the $i^{\textrm{th}}$ bit) that are crossed by
the cut. The total influence corresponds (after scaling) to the total
number of edges crossed.  Stated in this way, both the above
definitions have obvious extensions to the $k$-fold Cartesian product
of an arbitrary graph.

We consider the theorems of Kahn, Kalai, and Linial \cite{KahnKL88}
and that of Friedgut \cite{Friedgut98}, which are proved via Fourier
analysis for the hypercube, and generalize them to the case of
Cartesian products of arbitrary graphs.  We also consider applications
of these products to integrality gaps for linear and semidefinite
programming relaxations.

\paragraph{The KKL theorem.}
The theorem of Kahn, Kalai, and Linial \cite{KahnKL88}, which
introduced various tools in discrete Fourier analysis to Computer
Science, states that for the $k$-dimensional hypercube, a Boolean
function with variance $v$ has influence at least $\Omega(v \log k/k)$
along some coordinate.

We consider a generalization where we have a product graph $G^{\cart
  k}$, and the influence in the $i^{\textrm{th}}$ coordinate is
defined as the probability the function changes (i.e. one lands on the
other side of the of the cut) when taking a random step according to
$G$ in the $i^{\textrm{th}}$ direction. Informally, we prove the
following analogous statement for the case of Cartesian product of a
graph $G$.
\begin{theorem}[Generalized KKL Theorem, Informal]
  Given $f \from V(\Gk) \to \{-1,1\}$ with variance $v$, at least one
  of the coordinates coordinate has influence $\Omega\left(\alpha
    \cdot v \cdot (\log k)/k\right).$
\end{theorem}
Here, $\alpha$ is the log-Sobolev constant of the graph $G$, which is
a certain isoperimetric constant related to the mixing time for a
random walk on the graph.  We also discuss the tightness of these
results in Section~\ref{sec:tightness}. A similar KKL theorem for this
class of graphs was also obtained independently by Cordero-Erausquin
and Ledoux~\cite{Corder-EL11}, and a detailed comparison is included
later in this section. We also discuss several other results
generalizing the KKL theorem in different ways.  

\paragraph{Friedgut's junta theorem.} 
Friedgut's junta theorem states that if a Boolean function $f$ has
total influence $\I$, then there exists a Boolean function $g$
depending only on $\exp(O(\I/\e))$ coordinates such that
$\pr_x[f(x)\neq g(x)] \le \eps.$ Informally, we prove the following
analogous statement for the case of Cartesian products.
\begin{theorem}[Generalized Friedgut's Theorem, Informal]
Given $f \from V(\Gk) \to \{-1,1\}$ with total influence $\I$, there
exists a Boolean function $g$ depending only on $\exp(O(\I/\alpha\e))$
coordinates, such that $\pr_x[f(x)\neq g(x)] \le \eps.$
\end{theorem}
As before, $\alpha$ is the log-Sobolev constant of $G$.

\paragraph{Other generalizations of influences and related works.}
As mentioned before, the KKL theorem has been generalized in several
directions. In particular, Bourgain et al. \cite{BourgainKKKL}
considered the case when the function is a Boolean function over the
domain $[0,1]^n$. 
Mossel \cite{Mosselnotes} presents a proof for a function defined over
a product of $n$ finite probability spaces.
Keller et al., in~\cite{KellerMS12}, give a new definition of
influence in product spaces of continuous distributions, and prove
analogues to both the KKL theorem and Friedgut's theorem under the new
definition. They also explore several applications of these for
continuous and discrete probability spaces in \cite{KellerMS12b}.

In~\cite{Keller10}, Keller provides a family of definitions for
influence when the function is defined over the domain $[0,1]^n$, and
proves a generalization of the KKL theorem for these definitions. His
definitions consider a very general notion of influence of the
$i^{th}$ coordinate. As opposed to the usual notions, which define
influence by considering the variance in $f$ by changing the input in
the $i^{th}$ coordinate, his definition allows one to consider any
function $h$ of the \emph{expected value of the function over
  different inputs in the $i^{th}$ coordinate}\footnote{Note that this
  notion of influence is somewhat different from the one we
  consider. In our case, varying the input in the $i^{th}$ coordinate
  corresponds to varying the input over the copy of the graph $G$
  corresponding to the $i^{th}$ coordinate. While Keller's notion
  depends on the \emph{expected} value of the function over all the
  vertices of $G$, we are interested in how the function varies across
  the edges of $G$.}.

\Mnote{Added more text and a footnote comparing Keller's notion of influence to ours.}
 
Also, O'Donnell and Wimmer \cite{O'DonnellW09} obtained a
KKL theorem for a sub-class of Schreier graphs, which may not
necessarily have a product structure. 

The result most directly related to ours is the recent and independent work of
Cordero-Erausquin and Ledoux~\cite{Corder-EL11}. The authors prove a
generalization of the KKL theorem that, in particular, implies our
result for graph products, and the result of O'Donnell et
al.~\cite{O'DonnellW09} for certain classes of Schreier
graphs. However, to the best of our knowledge, their results do not
imply our generalization of the Friedgut's Junta theorem.

The proof in~\cite{Corder-EL11} builds upon the work
of~\cite{O'DonnellW09}, and is based on the hypercontractive
inequality for a Markov semigroup. The original proof of the
Friedgut's theorem is also based on hypercontractivity. In contrast,
our proofs are based on the log-Sobolev inequality, and generalize the
proofs of these theorems given by Rossignol \cite{Rossignol06} and by
Falik and Samorodnitsky \cite{FalikS07}.

\paragraph{Applications to integrality gaps.}
It follows from the results in \cite{HoudreT96} and \cite{ChungT98}
that if the starting graph $G$ has edge-expansion $h$, then the
product $G^{\cart k}$ has edge-expansion $h/k$. The same also holds
for the spectral gap of $G$ and $G^{\cart k}$, which is also the
optimum of the basic semidefinite program (SDP) for \SC. This
immediately implies that if one has a finite instance with integrality
gap $K$ for the basic SDP, then using Cartesian products it gives an
infinite family of arbitrarily large instances with the same gap.

We show that above is also the case for various hierarchies of linear
and semidefinite relaxations for \SC. In particular, if the optimum of
such a relaxation obtained by $r$ levels is $\textsf{Opt}$ for $G$,
then it is $\textsf{Opt}/k$ for $G^{\cart k}$. Most ways of increasing
the size of a graph seem to alter the expansion of the graph.
However, because of the above observation, Cartesian products provide
the right way of ``padding'' integrality gap instances to arbitrarily
large size while preserving the gap. We present the (simple) proofs of
these results in Appendix \ref{appendix:applications}.


%% file: prelims.tex
\section{Preliminaries and Notation}
For simplicity, for most of the paper, we will work with simple,
unweighted, regular graphs. All the results and proofs in the paper
can be extended to the case of general undirected graphs (or arbitrary
reversible Markov chains) by carefully picking the right
definitions. We provide the details in Section~\ref{sec:irregular}.

Let $[k]$ denote the set $\{1,\ldots,k\}.$ Given a graph $G,$ denote
its vertex set by $V(G),$ and its edge set by $E(G).$ First, we
formally define the Cartesian product of two graphs.
\begin{definition}[Cartesian product]
Given two graphs $G$ and $H$, their Cartesian product $G \cart H$ is
defined as a graph with the vertex set $V(G) \times V(H)$ and the
following set of edges,
\begin{align*}
E(G \cart H) \defeq \inbraces{ \inparen{(i_1,i_2),(j_1,j_2)} \suchthat
\insquare{\inparen{(i_1,j_1) \in E(G)} \wedge (i_2 = j_2) } \bigvee
\insquare{(i_1 = j_1) \wedge \inparen{(i_2,j_2) \in E(H)} } }.
\end{align*}
For a graph $G$, we define $G^{\cart 1} \defeq G$ and $G^{\cart k} \defeq
G^{\cart (k-1)} \cart G$.
\end{definition}
%
\Snote{Should we just define $\Gk$ instead?}

For the rest of this section, fix a positive integer $k,$ and let $G$
be a simple, unweighted, $d$-regular graph on $n$-vertices. Then, for
all $k,$ $\Gk$ is a simple, unweighted, $kd$-regular graph. The vertex
set of $\Gk$ is $V(G)^k$. For a vertex $x \in V(\Gk)$, we will use the
notation $x = (x_1,\ldots,x_k)$. We can think of each edge in $\Gk$ to
be along a \emph{coordinate} $j$, e.g. if $(x,y)$ is an edge, $x_i =
y_i$ for all $i \neq j$, we say $(x,y)$ is an edge along coordinate
$j$. We denote the set of all edges along coordinate $j$ as
$E_j(\Gk)$.

Let $\pi$ be the uniform distribution on $V(G)$. We define the inner
product and norms for the space of functions $V(G) \to \R$ as follows:
\[\pair{f,g} \defeq \av_{x\sim \pi} [f(x)g(x)] , \quad \norm{f}_2^2 \defeq
\pair{f,f} = \av_{x \sim \pi} [f(x)^2], \text{ and } 
\norm{f}_1 \defeq \Ex{x \sim \pi}{\abs{f(x)}} .\] 
The variance of a function $f: V(G) \to \R$ is defined with respect to
the same distribution, as 
\[\var(f) \defeq \Ex{x\sim \pi}{f(x)^2} -
\inparen{\Ex{x \sim \pi}{f(x)}}^2.\] Similarly, for the space of
functions $V(\Gk) \to \R,$ all the above notions are defined using the
uniform distribution over $V(\Gk),$ however we will use the same
notation for convenience, and the corresponding space of functions will
be clear from the context.

Let $\V{L_G}$ be the normalized Laplacian for the graph $G,$ $\V{L_G}
\defeq \V{I} - \frac{1}{d}\V{A},$ where $\V{I}$ denotes the identity
matrix, and $\V{A_G}$ denotes the combinatorial adjacency matrix of
$G.$ For any $f : V(G) \to \R,$ we have,
\begin{equation*}
\pair{f,\V{L_G}f} = \frac{1}{nd} \sum_{ (x,y) \in E(G)}
(f(x)-f(y))^2 = \frac{1}{2}\av_{(x,y) \in E(G)} (f(x)-f(y))^2 .
\end{equation*}
This immediately implies that $\V{L_G}$ is a positive semi-definite
operator.
For the graph $\Gk$, we also define the directional Laplacian
$\V{L_j}$ which only considers edges along the $j^{\textrm{th}}$ coordinate,
\begin{equation}
\label{eq:def-Lj}
\V{L_j} \defeq \V{I}\otimes \ldots \otimes \V{I}\otimes \V{L_G} \otimes
\V{I} \otimes \ldots \otimes\V{I},
\end{equation} which is a $k$-fold tensor with the matrix $\V{L_G}$ is in
the $j^\textrm{th}$ position. Thus, for any $f : V(\Gk) \to \R,$
\[\pair{f,\V{L_j}f} = \frac{1}{2}\av_{(x,y) \in E_{ j}(\Gk)}
(f(x)-f(y))^2 ,\]
where the average is taken over edges along coordinate $j.$ It is easy
to check that the Laplacian for $\Gk$ is $\V{L_{\Gk}} = \frac{1}{k}
\sum_j \V{L_j} = \av_j \V{L_j}$. The directional Laplacian also gives
the definition of influence for Cartesian products.
\begin{definition}[Influence]
For a boolean function $f: V(G^{\cart k}) \to \{-1,1\}$, we define its
influence along the $j^{\textrm{th}}$ coordinate as the quantity $\pair{f,
\V{L_j} f}$.
\end{definition}
Note that this definition is off by a factor
of 2 from the usual definition of influence for a boolean function on
the hypercube.  We also define the variance of the function along the
$j^{\textrm{th}}$ coordinate as below.  Here $x \setminus \{x_j\}$ denotes the
tuple $(x_1,\ldots,x_{j-1},x_{j+1},\ldots,x_k)$.
\begin{definition}[Variance along $j^\textrm{th}$ coordinate]
For a function $f: V(G^{\cart k}) \to \R$, its variance along the
$j^{\textrm{th}}$ coordinate is defined as $\varj{j}(f) \defeq \Ex{x \setminus
\{x_j\}}{\Ex{x_j}{f(x)^2} - (\Ex{x_j}{f(x)})^2}$.
\end{definition}
Letting $\V{J}$ denote the $n \times n$ matrix with all ones, define
the operator $\V{K_j}$ as the following $k$-tensor,
\begin{equation}
\label{eq:def-Kj}
\textstyle
\V{K_j} \defeq \V{I}\otimes \ldots \otimes \V{I} \otimes
\left(\V{I}-\frac{1}{n}\V{J}\right) \otimes \V{I} \otimes \ldots
\otimes\V{I},
\end{equation}
where the matrix $\V{I}-\frac{1}{n}\V{J}$ is in the $j^{\textrm{th}}$
position. We make the following simple observation.
\begin{claim}
For any $f : V(\Gk) \to \R$ and any $j \in [k],$ we have $\varj{j}(f)
= \pair{f, \V{K_j} f}$.
\end{claim}
\Snote{Do we need a proof of the claim? } For a boolean function on
the hypercube, $\pair{f,\V{L_j}f} = 2{\varj{j}(f)},$ and hence the two
definitions are essentially equivalent in this case. However, for a
general graph $G$, the variance does not depend of the structure of
the graph, whereas our notion of influence does.

\subsection{Isoperimetric Constants of a graph}

\paragraph{Conductance.}
Given a set $S \subseteq V(G)$, we define the volume of the set,
$\vol(S)$ to be the fraction of the vertices contained in $S$
i.e. $\vol(S) \defeq \nicefrac{|S|}{|V(G)|}$. We define the
Conductance of a graph $\Phi(G)$, as follows
\[\Phi(G) \defeq \min_{\small \substack{S \subset V(G) \\ S \neq
\emptyset, V(G)}} \frac{1}{4}\frac{|E(S,\bar{S})|}{|E|}
\frac{1}{\vol(S)\vol(\bar{S})}.\] The factor of $\nicefrac{1}{4}$
ensures that $\Phi(G) \le 1$.  If we consider the $\{-1,1\}$-valued
indicator function of a set $S$, we get an equivalent definition of
$\Phi(G)$ as follows,
\[\Phi(G) = \min_{\small \substack{f \from V(G) \to \{-1,1\} \\
    \var(f) \neq 0}}
\frac{\pair{f,\V{L_G}f}}{2\var(f)}.\]
This implies that if $f$ is any $\{-1,1\}$-valued function,
\[\pair{f,\V{L_G}f} ~\ge~ 2 \Phi(G) \cdot \var(f).\] Since
$\V{L_G}$ is positive semi-definite, this holds even if $\var(f) = 0.$
This is also true if $f$ is a $\B$-valued indicator function for a
set. 

We know that $\Gk$ consists of several copies of $G$ along each
coordinate. Hence, a similar statement holds for the directional
Laplacians $\V{L_j},$ as proved in the following straightforward
lemma.
\begin{lemma}
\label{lem:Lj-conductance}
For any $f : V(\Gk) \to \{-1,1\},$ and $j \in [k],$ we have
$\pair{f,\V{L_j}f} \ge 2\Phi(G)\cdot \varj{j}(f).$
\end{lemma}
\begin{proof}
Without loss of generality, let $j=1.$ If we fix $x\backslash \{x_1\}
\in V(G)^{k-1},$ and consider all $x_1 \in V(G)$, we get a copy of
$G.$ Denote the restriction of $f$ to these vertices as
$f_{x\backslash \{x_1\}}$. On this copy of the graph, we know that,
\[\pair{f_{x\backslash \{x_1\}},\V{L_G}f_{x\backslash \{x_1\}}} \ge 2 \Phi(G)\cdot
\var_{x_1}(f_{x\backslash \{x_1\}}) . \]
Averaging the above equation over all $x\backslash \{x_1\}$, we get that,
\[\pair{f,\V{L_1}f}  = \av_{x\backslash \{x_1\}} \pair{f_{x\backslash
    \{x_1\}},\V{L_G}f_{x\backslash \{x_1\}}} \ge 2\Phi(G)\cdot
\av_{x\backslash \{x_1\}} \var_{x_1}(f_{x\backslash \{x_1\}}) = 2\Phi(G) \cdot
\varj{1}(f) \vspace{-11pt}.\]
\end{proof}
\paragraph{Eigenfunctions and Eigenvalues.}
Since $\V{L_G}$ is a symmetric operator (more precisely, it is
self-adjoint on the space of functions under consideration), it has
real eigenvalues $\lambda_0 = 0 \le \lambda_1 \le \ldots \le
\lambda_{n-1}.$ We fix a basis of eigenfunctions $v_0 =
\mathbbm{1},v_1,\ldots,v_{n-1},$ that is orthonormal, \emph{i.e.},
$\pair{v_i,v_j} = 1$ iff $i = j$ and 0 otherwise, and such that
$\V{L_G}v_i = \lambda_i v_i.$ It is well known that the eigenfunctions
of $\V{L_{\Gk}}$ are tensor products of the eigenfunctions of $\V{L_G}$.
\begin{proposition}[Eigenvalues, Eigenfunctions]
\label{prop:Gk-eigen}
Let $v_0,\ldots,v_{n-1}$ be an orthonormal basis of eigenfunctions for
$\V{L_G},$ with eigenvalues $\lambda_0,\ldots,\lambda_{n-1}.$ Then,
for every positive integer $k,$ the set of vectors $\{v_{i_1}\otimes
\ldots\otimes v_{i_k}\}_{i_1,\ldots,i_k \in \{0,\ldots,n-1\}}$ form an
orthonormal basis of eigenvectors for $\V{L_{\Gk}},$ where
$v_{i_1}\otimes \ldots\otimes v_{i_k}$ has eigenvalue $\av_j
\lambda_{i_j}.$
\end{proposition}
\begin{proof}
Fix a sequence $i_1,\ldots,i_k,$ and the vector $v_{i_1}\otimes
\ldots\otimes v_{i_k}.$ Thus,
\begin{align*}
\V{L_{\Gk}} (v_{i_1}\otimes \ldots\otimes v_{i_k}) & =
\left( \av_j \V{L_j} \right) (v_{i_1}\otimes \ldots\otimes v_{i_k})  =
\av_j \left( \V{L_j} (v_{i_1}\otimes \ldots\otimes   v_{i_k}) \right)\\
& = \av_j (v_{i_1} \otimes \ldots \otimes \V{L_G} v_{i_j} \otimes \ldots
\otimes v_{i_k}) \\
& = \av_j (v_{i_1} \otimes \ldots \otimes \lambda_{i_j} v_{i_j} \otimes
\ldots \otimes v_{i_k}) = \left( \av_j \lambda_{i_j} \right) v_{i_1}
\otimes \ldots \otimes v_{i_k}.
\end{align*}
Thus, $v_{i_1}\otimes \ldots \otimes v_{i_k}$ is an eigenfunction with
eigenvalue $\av_j \lambda_{i_j}.$ Moreover, 
\[\pair{v_{i_1} \otimes \ldots \otimes v_{i_k}, v_{l_1}\otimes \ldots
  \otimes v_{l_k}} = \pair{v_{i_1},v_{l_1}} \ldots
\pair{v_{i_k},v_{l_k}},\]
which is 0 unless $i_j = l_j$ for all $j,$ in which case it is 1. Thus
they are orthonormal, and by a dimensionality argument, they form a
basis.
\end{proof}

Letting $(i)$ denote the sequence $(i_1,\ldots,i_k)$, we denote the
eigenfunction $v_{i_1}\otimes v_{i_2}\otimes\ldots\otimes v_{i_k}$ by
$v_{(i)}$. 
%
%
%
\paragraph{Log-Sobolev Constant.} For a function $f : V(G) \to \R$, we
define the entropy of the function as follows,
\begin{align*}
\Ent(f^2) &~\defeq~ \av_{x \sim \pi} [f(x)^2\log f(x)^2] - (\av_{x
\sim \pi} [f(x)^2])\log \av_{x \sim \pi} [f(x)^2]\\
&~=~ \av_{x \sim \pi} [f(x)^2\log f(x)^2] - \norm{f}_2^2\log \norm{f}_2^2.
\end{align*}
where $\log$ is the natural logarithm.
\begin{definition}[Log-Sobolev Constant]
The log-Sobolev constant of a graph $G$ is defined to be the largest
constant $\alpha(G)$ \footnote{Often in the literature, \emph{e.g.}
in~\cite{DiaconisS96}, the log-Sobolev constant is defined to be twice
the definition we use.} such that the following inequality holds for
all functions $f : V(G) \to \R$,
\begin{equation}
\label{eq:log-sob}
\pair{f,\V{L_G}f} ~=~ \frac{1}{2} \av_{(x,y) \in E(G)} (f(x)-f(y))^2
~\ge~ \frac{\alpha(G)}{2} \cdot \Ent(f^2).
\end{equation}
\end{definition}
The above inequality is called the \emph{log-Sobolev inequality} for
graph $G.$ The following lemma relates the log-Sobolev constant of
$\Gk$ to that of $G$.
\begin{lemma}[Lemma 3.2, Diaconis and Saloff-Coste \cite{DiaconisS96}]
\label{lem:alpha-Gk}
Let $\alpha(G)$ be the log-Sobolev constant for a graph $G$, then the
log-Sobolev constant for $\Gk$ is $\alpha(G)/k$.
\end{lemma}
It is known that the isoperimetric constants defined above satisfy the
following inequalities between them (see Lemma 3.1 in
\cite{DiaconisS96} for example),
\[\alpha(G) ~\le~ \lambda_1(G) ~\le~ 2\Phi(G).\]


%% file: theorems.tex
\section{KKL Theorem and Friedgut's Junta Theorem}
In this section, we shall prove the analogues of the theorems of Kahn,
Kalai, and Linial \cite{KahnKL88}, and Friedgut \cite{Friedgut98} for
Cartesian products of graphs. Both these theorems analyze cuts in the
hypercube which is simply the Cartesian product of an edge. The proofs
of both theorems proceed by using hypercontractivity of the
Bonami-Beckner noise operator on the hypercube.

While the noise operator can be easily generalized to the setting of
Cartesian products, the hypercontractivity based proofs do not seem to
extend easily to the setting of Cartesian products of general
graphs. Instead we develop on Rossignol's proof of the KKL theorem
\cite{Rossignol06}, which is based on the log-Sobolev inequality. For
the KKL and Friedgut theorems on the hypercube, proofs using the
log-Sobolev inequality were also given by Falik and Samorodnitsky
\cite{FalikS07}.

To prove both the theorems, we shall need some preparatory lemmas. We
develop these below. The manipulations are similar to those in
\cite{Rossignol06}. For the rest of this section, fix a simple,
regular and unweighted graph $G,$ and a positive integer $k.$ All the
results in this section hold for any such $G$ and $k.$


Let $f : V(\Gk) \to \{-1,1\}$ define a cut in the graph $\Gk$. Let $f$
be represented in the basis of the eigenfunctions of $\V{L_{\Gk}}$ as $f =
\sum_{(i)} \widehat{f}_{(i)} v_{(i)}$ where $\widehat{f}_{(i)} \defeq
\pair{f,v_{(i)}}$. 
\begin{lemma}
\label{lem:vj-var-comparison}
For any $f : V(\Gk) \to \R,$ with the eigenbasis representation $f =
\sum_{(i)} \widehat{f}_{(i)} v_{(i)},$ we have $\var(f) = \sum_{(i)
  \neq 0} \widehat{f}_{(i)}^2, \textrm{ and } \varj{j}(f) =\sum_{(i):
  i_j \neq 0} \widehat{f}_{(i)}^2.$ In particular, this implies
\footnote{The second half of the inequality is the well-known
  Efron-Stein inequality (see \cite{Steele86} for example).}
\begin{equation}
\label{eq:vj-var-comparison}
\textstyle
\max_j \varj{j}(f) \le \var(f) \le \sum_{j \in [k]} \varj{j}(f).
\end{equation}
\end{lemma}
\begin{proof}
We have $f = \sum_{(i)} \widehat{f}_{(i)} v_{(i)}.$ Thus,
\[\textstyle \var(f) = \av_{x}{f(x)^2} - \left(\av_{x}f(x)\right)^2 =
\norm{f}_2^2 - \pair{v_{(0)},f}^2 = \sum_{(i)} \widehat{f}_{(i)}^2 -
\widehat{f}_{(0)}^2 =\sum_{(i) \neq 0} \widehat{f}_{(i)}^2, \] where
we have used the fact that $\{v_{(i)}\}_{(i)}$ is an orthonormal
basis, and $v_{(0)} = \mathbbm{1}.$ Also, observing
that for any $j \in [k]$ and any tuple $(i),$ $\V{K_j} v_{(i)} = v_{(i)}$
iff $i_j \neq 0$ and $0$ otherwise,
\[\varj{j}(f) = \pair{f,\V{K_j}f} = \pair{\sum_{(i)}
  \widehat{f}_{(i)}v_{(i)}, \sum_{(i)}
  \widehat{f}_{(i)} \V{K_j}v_{(i)} } = \pair{\sum_{(i)}
  \widehat{f}_{(i)}v_{(i)}, \sum_{(i): i_j \neq 0}
  \widehat{f}_{(i)}v_{(i)} } = \sum_{(i): i_j \neq 0}
\widehat{f}_{(i)}^2.\]
From these representations, Equation~\eqref{eq:vj-var-comparison}
follows immediately.
\end{proof}

For $j \in [k],$ define the functions $f_j$ as follows:
\[f_j ~\defeq~ \sum_{(i):i_j \neq 0, i_l = 0\ \forall l > j}
\widehat{f}_{(i)} v_{(i)} ~=~ \Ex{x_{j+1},\ldots,x_k} {\V{K_j} f}
.\] 
\begin{lemma}[Basic Properties of $\{f_j\}_{j \in [k]}$]
\label{lem:fj-prop}
  The functions $\{f_j\}_{j \in [k]}$ defined above satisfy:
\begin{enumerate}[itemsep=0pt]
\item For $j_1,j_2 \in [k]$ such that $j_1 \neq j_2,$
  $\pair{f_{j_1},f_{j_2}} = \pair{f_{j_1},\V{L_{\Gk}}f_{j_2}} = 0.$
\item $ \sum_{j \in [k]} \pair{f_j,\V{L_{\Gk}}f_j}  =
  \pair{f,\V{L_{\Gk}}f}.$
\item $\sum_{j \in [k]} \norm{f_j}_2^2 = \var(f).$
\end{enumerate}
\end{lemma}
\begin{proof}
We note that the functions $\{f_j\}_{j \in [k]}$ are projections of
$f$ onto orthogonal eigenspaces of $\V{L_{\Gk}}$. This implies that 
these functions are orthogonal, i.e. $\pair{f_{j_1},f_{j_2}} = 0$ for
$j_1 \neq j_2,$ and also $\pair{f_{j_1},\V{L_{\Gk}}f_{j_2}} = 0$ for
$j_1 \neq j_2$. This also implies,
\[\textstyle \sum_{j} \pair{f_j,\V{L_{\Gk}}f_j} = \pair{   \sum_{j }
  f_j ,\V{L_{\Gk}} \sum_{j} f_j } = \pair{  \widehat{f}_{(0)}v_{(0)}+
  \sum_{j } f_j ,\V{L_{\Gk}} \left( \widehat{f}_{(0)}v_{(0)} +
    \sum_{j} f_j \right)} = \pair{f,\V{L_{\Gk}}f},\]
where we have used the fact that $v_{(0)}$ is orthogonal to all $f_j,$
and $\pair{v_{(0)},\V{L_{\Gk}}v_{(0)}} = 0.$  

Finally, using the definition of $f_j's$ and the fact that
$\{v_{(i)}\}_{(i)}$ are orthonormal, 
\[\sum_{j \in [k]} \norm{f_j}_2^2 = \sum_{j \in [k]} \pair{f_j,f_j} =  \sum_{j \in [k]} \sum_{(i):i_j
  \neq 0, i_l = 0\ \forall l > j} \widehat{f}_{(i)}^2 = \sum_{(i)
  \neq 0} \widehat{f}_{(i)}^2 = \var(f),\]
by Lemma~\ref{lem:vj-var-comparison}. 
\end{proof}

Next, we bound the norms of the functions $f_j$.
\begin{lemma}[Norm Bounds for $f_j$]
\label{lem:norm-bounds}
For all $j \in [k],$ the $\ell_2$ and $\ell_1$ norms of the functions
$f_j$ are bounded as follows: $\norm{f_j}_2^2 ~\le~ \varj{j}(f) ~=~
\pair{f,\V{K_j}f},$ and $ \norm{f_j}_1 ~\le~ \varj{j}(f).$
\end{lemma}
\begin{proof}
For the first part, 
\[\norm{f_j}_2^2  = \sum_{(i):i_j \neq 0, i_l = 0\ \forall l > j}
\widehat{f}^2_{(i)} \le \sum_{(i):i_j \neq 0} \widehat{f}^2_{(i)} = \varj{j}(f)
= \pair{f,\V{K_j}f}.\]

For the second part, we start with the triangle inequality to upper
bound $\abs{f_j (x)}$
\begin{align*}
 \norm{f_j}_1 = \Ex{x}{\abs{f_j(x)}} = \Ex{x}{\abs{\Ex{x_{j+1},\ldots,x_k}{\V{K_j}
f(x)}}}
& ~\le~ \Ex{x} {\abs{\V{K_j} f(x)}} \\
& ~=~ \Ex{x} {\abs{f(x)-\Ex{y ~:~ y_i = x_i \forall i \neq j }{ f(y)} } } \\
& ~=~ \Ex{x} {f(x) \cdot \inparen{f(x)-\Ex{y ~:~ y_i = x_i \forall i
\neq j}{f(y)} }}\\
& ~=~ \pair{f,\V{K_j}f} ~=~ \varj{j}(f),
\end{align*}
where we used the observation that the sign of $f(x)-\av_{y} f(y)$
is the same as $f(x)$, since $f(x) \in \{-1,1\},$ and $\abs{\Ex{y ~:~
    y_i = x_i \forall i \neq j }{ f(y)}} \le 1.$
\end{proof}

We shall apply the log-Sobolev inequality to the functions $f_j$
defined above.  However, the entropy of these functions is somewhat
difficult to work with.  The following lemma gives a different
estimate in terms of the $\ell_1$ and $\ell_2$ norms of the functions
we are applying the log-Sobolev inequality to.
\begin{lemma}
For any $t \in (0,\nfrac{1}{e^2}]$ and $h: V(\Gk) \to \R$,
\begin{equation*}
\pair{h,\V{L_{\Gk}}h} ~\ge~ \frac{\alpha(G)}{2k} \cdot \left(\sqrt{t}\log
t \cdot \norm{h}_1 + \log t \cdot \norm{h}_2^2 -
\norm{h}_2^2\log\norm{h}_2^2\right).
\end{equation*}
\end{lemma}

\begin{proof}
Applying the log-Sobolev inequality (Equation~\eqref{eq:log-sob}) for $\Gk$ to $h$,
\begin{equation}
\label{eq:simplify-entropy:log-sob}
\pair{h,L_{\Gk}h} ~\ge~ \frac{\alpha(\Gk)}{2}\Ent(h^2) ~=~
\frac{\alpha(G)}{2k} \cdot \inparen{\Ex{x}{ h(x)^2\log h(x)^2} -
  \norm{h}_2^2\log \norm{h}_2^2},
\end{equation} where we used the definition of
$\Ent(\cdot),$ and the fact that $\alpha(\Gk) = \alpha(G)/k,$ from
Lemma~\ref{lem:alpha-Gk}.

Observe that since $t \in (0,\nfrac{1}{e^2}],$ the function $\sqrt{z}
\log z$ is decreasing in $[0,t].$ We use this to bound the first term
as below.
\begin{align*}
\Ex{x} {h^2(x)\log h(x)^2} 
&~=~ \Ex{x}{h^2(x)\log h(x)^2 \cdot \V{1}_{h^2 \le t}} + \Ex{x}{
h^2(x)\log h(x)^2 \cdot \V{1}_{h^2 > t}} \\
&~=~ \Ex{x}{\abs{h(x)} \cdot \sqrt{h(x)^2}\log h(x)^2 \cdot
\V{1}_{h^2 \le t}} + \Ex{x}{
h^2(x)\log h(x)^2 \cdot \V{1}_{h^2 > t}} \\
&~\ge~ \Ex{x}{ |h(x)| \cdot \sqrt{t}\log t \cdot \V{1}_{h^2 \le t}} +
\Ex{x}{h^2(x) \cdot \log t \cdot \V{1}_{h^2 > t}}\\
&~\ge~ \sqrt{t}\log t \cdot \Ex{x} {|h(x)|} + \log t \cdot \Ex{x}{ h^2(x)} \\
&~=~ \sqrt{t}\log t \cdot \norm{h}_1 + \log t \cdot \norm{h}_2^2,
\end{align*}
where the last inequality used the fact that $\log t < 0$ for $t \in
(0,\nfrac{1}{e^2}]$. Plugging the above bound in the inequality from
Equation~\eqref{eq:simplify-entropy:log-sob} proves the claim.
\end{proof}
Combining the above lemma with Lemma~\ref{lem:norm-bounds} gives the
following corollary, which shall be useful in the proofs of both the
theorems.
\begin{corollary} \label{cor:main}
For all $t \in \left(0,\nfrac{1}{e^2}\right.]$ and all $j \in [k],$
\begin{align*}
  \pair{f_j,\V{L_{\Gk}}f_j} & ~\ge~ \frac{\alpha(G)}{2k} \cdot
  \left(\sqrt{t}\log t \cdot \varj{j}(f) + \log t \cdot \norm{f_j}_2^2
    - \norm{f_j}_2^2\log \norm{f_j}_2^2\right).
\end{align*}
\end{corollary}

\subsection{KKL Theorem for Cartesian products}
We now prove the following analogue of the KKL theorem, which says
that a cut with high variance must have a somewhat large number of
edges crossing in at least one direction.
\begin{theorem}[Generalized KKL Theorem]
\label{thm:kkl}
Given $f \from V(\Gk) \to \{-1,1\}$, we have
\[
\max_j \pair{f,\V{L_j}f} ~\ge~ \Omega\left(\frac{\log k}{k} \cdot
  \alpha(G) \cdot \var(f)\right).\]
\end{theorem}
Note that in comparison with the usual KKL theorem for the hypercube,
our bound has an extra factor of $\alpha(G)$, which is $2$ for the
case of the hypercube where $G$ is an edge.
\begin{proof}
Let $\alpha \defeq \alpha(G)$ and $\Phi \defeq \Phi(G).$ Let $V =
\max_j \varj{j}(f)$. Using Lemma~\ref{lem:norm-bounds},
$\norm{f_j}_2^2 \le \varj{j}(f) \le V.$ Plugging this into
Corollary~\ref{cor:main}, we get that for all $j \in [k]$ and $t \in
(0,\nfrac{1}{e^2}],$
\begin{align*}
  \pair{f_j,\V{L_{\Gk}}f_j} & ~\ge~ \frac{\alpha}{2k} 
  \left(V\cdot \sqrt{t}\log t + \log t \cdot \norm{f_j}_2^2
    - \norm{f_j}_2^2\log V\right),
\end{align*}
where we have used $\log t < 0$ for $t \in (0,\nfrac{1}{e^2}].$ Adding
the above equation for all $j \in [k],$ and using $\sum_{j \in [k]}
\pair{f_j,\V{L_{\Gk}}f_j} = \pair{f,\V{L_{\Gk}}(f)}$ and $\sum_{j \in
  [k]} \norm{f_j}_2^2 = \var(f)$ from Lemma~\ref{lem:fj-prop}, we get
for any $t \in (0,\nfrac{1}{e^2}],$
\[\pair{f,\V{L_{\Gk}}f} ~\ge~ \frac{\alpha}{2k}\left(kV \cdot \sqrt{t}\log t
  + \log t \cdot \var(f) - \var(f)\log V\right).\]
In order to balance the two expressions involving $t$ in the above
equation, and still have $t \le \frac{1}{e^2}$, we pick $t =
\left(\frac{\var(f)}{ekV}\right)^2$ (since $\var(f) \le \sum_{j \in
  [k]} \varj{j}(f) \le kV$, using
Equation~\eqref{eq:vj-var-comparison}). This gives,
\[\pair{f,\V{L_{\Gk}}f} ~\ge~ \frac{\alpha\var(f)}{2k} \cdot
\left(2\left(1+\frac{1}{e}\right)\log \left(\frac{\var(f)}{ekV}\right)
  + \log \frac{1}{V}\right).\]
Suppose that $V \ge \frac{\log k}{k} \cdot \var(f)$. Then, using
$\pair{f,\V{L_j}f} \ge 2\Phi\cdot \varj{j}(f)$ from
Lemma~\ref{lem:Lj-conductance}, we get that,
\[\textstyle \max_j \pair{f,\V{L_j}f}
\ge 2\Phi \cdot \max_j \varj{j}(f) = 2\Phi V
= \Omega\left(\frac{\log k}{k} \cdot \Phi \cdot \var{f}\right).\]
Also, for the case when $V \le \frac{\log k}{k} \cdot \var{f}$, we
have,
\begin{align*}
\pair{f,\V{L_{\Gk}}f} & ~\ge~ \frac{\alpha\var(f)}{2k} \cdot
\left(2\left(1+\frac{1}{e}\right)\log \left(\frac{1}{e\log k}\right) +
\log \left(\frac{k}{\var(f)\log k}\right)\right) \\
 &~=~ \Omega\left(\frac{\log k}{k} \cdot \alpha \cdot \var(f)\right).
\end{align*}
Combining the above with the facts that $2\Phi \geq \alpha$ and
$\max_j \pair{f,\V{L_j}f} \ge \av_j \pair{f,\V{L_j}f} =
\pair{f,\V{L_{\Gk}}f}$ then proves the result.
\end{proof}
\subsection{Friedgut's theorem for Cartesian products of graphs}
Friedgut's theorem says that a cut on the hypercube which is crossed
by very few edges is close to a cut that depends only on a few
coordinates. We now prove the following analogue of Friedgut's
theorem. 
\begin{theorem}[Generalized Friedgut's Junta Theorem]
\label{thm:friedgut}
Given any $f \from V(\Gk) \to \{-1, 1\}$, $f$ is $\epsilon$-close to a
boolean function $g : V(\Gk) \to \{-1, 1\}$, i.e. $\norm{f-g}_2^2 \le
\epsilon$, which is determined only by the value of $l$ coordinates,
where,
\[l ~~\le~~ \exp\left(\frac{50k}{\eps}\frac{1}{\alpha(G)} \cdot
  \pair{f,\V{L_{\Gk}}f}\right)\]
\end{theorem}
As before, substituting $\alpha(G) = 2$ gives the hypercube
version. Also note that since $f$ and $g$ are both Boolean functions,
$\pr_{x\sim \pi}[f(x) \neq g(x)] = \frac{1}{4}\cdot \Ex{x \sim
  \pi}{(f(x)-g(x))^2} = \frac{1}{4} \cdot \norm{f-g}_2^2 \le
\frac{1}{4} \cdot \eps.$

One interesting way to interpret the generalization of Friedgut's
Junta theorem is as follows. Suppose that for a family of graphs $G$,
$\Phi$ and $\alpha$ are within a constant factor (as a function of the
size of $G$). Also, we know that the cut with the least conductance in
$\Gk$ has value $\Phi/k$ and depends only on a single coordinate (see
Appendix~\ref{appendix:applications}). The above theorem says that a
balanced cut in $\Gk$ (described by $f$) that has sparsity
$O(\Phi/k)$, must be $\epsilon$-close to a cut that is determined only
by a constant number of coordinates.
\begin{proof}
Let $\alpha \defeq \alpha(G)$ and $\Phi \defeq \Phi(G).$ We order the
coordinates $j$ are so that $\varj{j}(f)$ is non-increasing.  Let $J =
\{1, \ldots, l\}$ be the subset of all coordinates $j$ that have
$\varj{j}(f)$ at least $V$ (where $V$ is some threshold we will pick
later).  Let $f = \sum_{(i)} \widehat{f}_{(i)} v_{(i)}$.  We define
the function $g : V(\Gk) \to \R$ as follows,
\[g(x) = \sum_{(i) : i_j = 0\ \forall j \notin J} \widehat{f}_{(i)}
v_{(i)}.\] We know that if $i_j=0,$ then $v_{(i)}$ does not depend on
the $j^{\textrm{th}}$ coordinate. Thus, $g$ only depends on coordinates in
$J$. We shall show that for an appropriate choice of $V$,
$\norm{f-g}_2^2$ is small i.e. $\norm{f-g}^2 \le \epsilon$. It will
also follow from our choice of $V$ that the number of coordinates in
$J$ is as claimed.

For any $j \notin J,$ using Lemma~\ref{lem:norm-bounds},
$\norm{f_j}_2^2 \le \varj{j}(f) \le V.$ Plugging this into
Corollary~\ref{cor:main}, we get that for all $j \notin J$ and $t \in
(0,\nfrac{1}{e^2}],$
\begin{align*}
\pair{f_j,\V{L_{\Gk}}f_j} & \ge \frac{\alpha}{2k} \left(\varj{j}(f)
\cdot \sqrt{t}\log t + \norm{f_j}_2^2 \cdot \log t - \norm{f_j}_2^2
\cdot \log V\right),
\end{align*}
where we have used $\log t < 0$ for $t \in (0,\nfrac{1}{e^2}].$ Adding
the above equation for all $j \notin J$ and observing $\sum_{j \notin
  J} \pair{f_j,f_j} = \norm{f-g}_2^2,$ we get that for any $t \in
(0,\nfrac{1}{e^2}]$,
\begin{align} \label{eq:friedgut_proof}
\sum_{j \notin J} \pair{f_j,\V{L_{\Gk}}f_j} & ~\ge~
\frac{\alpha}{2k} \cdot \left(\sqrt{t}\log t \sum_{j \notin
    J} \varj{j}(f) + \norm{f-g}_2^2 \cdot \log t - \norm{f-g}_2^2
  \cdot \log V\right).
\end{align}
Using the fact that $\V{L_{\Gk}}$ is positive semi-definite, and
$\sum_{j \in [k]} \pair{f_j,\V{L_{\Gk}}f_j} = \pair{f,\V{L_{\Gk}}f}$
from Lemma~\ref{lem:fj-prop}, we can bound the LHS as
\[\textstyle
\sum_{j \notin J}\pair{f_j,\V{L_{\Gk}}f_j} ~\le~ \sum_{j \in [k]}
\pair{f_j,\V{L_{\Gk}}f_j} = \pair{f,\V{L_{\Gk}}f}.\]
From Lemma~\ref{lem:Lj-conductance}, we know that $\pair{f, \V{L_j} f}
\geq 2\Phi \cdot \varj{j}(f)$. This gives a bound on $\sum_{j \notin
  J} \varj{j}(f)$ in terms of $\pair{f,\V{L_{\Gk}f}}$.
\begin{equation}
\label{eq:friedgut-proof:var-bound}
\textstyle
\sum_{j
  \notin J} \varj{j}(f) ~\le~ \sum_{j \in [k]} \varj{j}(f) ~\le~ 
\frac{1}{2\Phi} \cdot \pair{f, \sum_{j \in [k]} \V{L_j} f}
~=~ \frac{k}{2\Phi} \cdot \pair{f,\V{L_{\Gk}}f}. 
\end{equation}
We now plug in the above bounds into equation
(\ref{eq:friedgut_proof}). Remember that $\log t < 0$. Thus, we get
for any $t \in (0,\nfrac{1}{e^2}]$,
\begin{equation*}
\textstyle
\pair{f,\V{L_{\Gk}}f} ~\ge~ \frac{\alpha}{2k} \cdot
\left(\frac{k}{2\Phi}\sqrt{t}\log t \cdot \pair{f,\V{L_{\Gk}}f} +
\norm{f-g}_2^2 \cdot \log t - \norm{f-g}_2^2 \cdot \log V\right).
\end{equation*}
Again, in order to balance the terms and still have $t \le
\frac{1}{e^2}$, we choose $t = \left(\frac{2\Phi \cdot
    \|f-g\|^2}{ek\cdot \pair{f,\V{L_{\Gk}}f}}\right)^2$ (since
$\norm{f-g}_2^2 \le \var(f) \le \sum_{j \in [k]} \varj{j}(f) \le
\frac{k}{2\Phi} \pair{f,\V{L_{\Gk}}f},$ using
Equations~\eqref{eq:vj-var-comparison} and
\eqref{eq:friedgut-proof:var-bound}).  For this value of $t$, the
above bound simplifies to
\begin{align*}
\pair{f,\V{L_{\Gk}}f} & ~\ge~ \frac{\alpha}{2k} \cdot \|f-g\|^2 \cdot
\left(2\left(1+\frac{1}{e}\right)\log \left(\frac{2\Phi \cdot
\norm{f-g}^2}{ek \cdot \pair{f,\V{L_{\Gk}}f}}\right)- \log V\right) \\
\Longrightarrow \qquad\qquad\log V & ~\ge~
2\left(1+\frac{1}{e}\right)\log \left(\frac{2\Phi \cdot
\norm{f-g}^2}{ek \cdot \pair{f,\V{L_{\Gk}}f}}\right) - \frac{2k \cdot
\pair{f,\V{L_{\Gk}}f}}{\alpha \cdot \norm{f-g}^2} .
\end{align*}
Consider the RHS as a function in $\norm{f-g}^2$, say
$F(\norm{f-g}^2)$ and note that $F$ is an increasing function. Thus,
choosing $\log V = F(\epsilon)$ and hence $V = \exp(F(\epsilon))$
would imply that $\norm{f-g}^2 \leq \epsilon$. It only remains to show
that the size of the set $J$ is bounded as claimed, for this choice of
$V$.

Since $J$ was defined to be the set of coordinates $j$ with
$\varj{j}(f) \geq V$, the size of $J$ is at most $(\sum_j
\varj{j}(f))/V = (\sum_j \varj{j}(f)) \cdot \exp(-F(\epsilon))$.
Using the bounds $\alpha \leq 2\Phi,$ and $\sum_j \varj{j}(f) \leq
(k/2\Phi) \cdot \pair{f, \V{L_{\Gk}} f}$ from
Equation~\eqref{eq:friedgut-proof:var-bound}, we can bound this as
\begin{align*}
|J| & ~\le~ \left(\sum_j \varj{j}(f)\right) \exp\left(
\frac{2k}{\epsilon\alpha} \cdot \pair{f,\V{L_{\Gk}}f} +
2\left(1+\frac{1}{e}\right) \cdot \log \left(\frac{ek\cdot
\pair{f,\V{L_{\Gk}}f}}{2\Phi\epsilon}\right) \right) \\
& ~\leq~ \exp\left( \log\left( \frac{k \cdot \pair{f, \V{L_{\Gk}}
f}}{2\Phi}\right) + \frac{2k}{\epsilon\alpha} \cdot
\pair{f,\V{L_{\Gk}}f} + 2\left(1+\frac{1}{e}\right) \cdot \log
\left(\frac{ek\cdot \pair{f,\V{L_{\Gk}}f}}{2\Phi\epsilon}\right)
\right) \\
& ~\le~ \exp\left( \frac{2k \cdot
\pair{f,\V{L_{\Gk}}f}}{\epsilon\alpha} + \left(3+\frac{2}{e}\right)
\log \left(\frac{2k \cdot \pair{f,\V{L_{\Gk}}f}}{\Phi\epsilon} \right)
\right) \\
& ~\le~ \exp\left(\frac{12 k}{\epsilon\alpha} \cdot
\pair{f,\V{L_{\Gk}}f}\right).
\end{align*}
Thus we have a function $g : V(\Gk) \to \R$ that depends only on a
\emph{few} coordinates and is \emph{close} to $f$. Now, we prove that
the boolean function obtained by taking the sign of $g$, denoted
$\tilde{g} = sgn(g) : V(\Gk) \to \{-1,1\}$ also satisfies $\|f -
\tilde{g}\|^2 = O(\epsilon)$. To see this, observe that whenever $f(x)
\neq \tilde{g}(x)$, $|f(x)-\tilde{g}(x)|^2 = 4$ but $|f(x)-g(x)|^2$
must be at least 1. Hence $\|f-\tilde{g}\|^2 \le 4\epsilon$. Replacing
$\epsilon$ by $\epsilon/4$ proves the theorem.
\end{proof}

%% file: tightness.tex
\section{Tightness of KKL theorem for Cartesian Products of Graphs}
\label{sec:tightness}

\paragraph{A KKL theorem in terms of conductance.}
Our generalization of the KKL theorem for Cartesian products of graphs
in particular implies one for the $q$-ary hypercube $[q]^k$.  It is
easy to see that the $q$-ary hypercube is exactly the graph
$K_q^{\cart k}$ where $K_q$ denotes the complete graph. We need the
log-Sobolev constant of $K_q$, which is known to be
$\Theta(\frac{1}{\log q})$ for $q \ge 3$ \cite{DiaconisS96}.  Also,
the (normalized) Laplacian for the complete graph on $q$ vertices,
without self-loops, is
\[ \V{I} - \frac{1}{q-1} \cdot (\V{J} - \V{I}) ~=~ \frac{q}{q-1} \cdot
\inparen{\V{I} - \frac{1}{q} \cdot \V{J}}
\]
This gives the following as an easy corollary.
\begin{corollary}
Given a function $f : V(K_q^{\otimes k}) \to \{-1, 1\}$,
\[\textstyle \max_j \pair{f,\V{K_j} f} ~\ge~ \Omega\left(\frac{\log k}{k} \cdot
\frac{\var(f)}{\log q}\right)\]
\end{corollary}

Assuming this result, there is a simple way to obtain a variant of the
KKL with different bounds. We can apply the conductance bound for
graph $G$ to conclude,
\begin{corollary}
\label{cor:kkl-qary}
Given a function $f : V(\Gk) \to \{-1,1\}$, 
\[\textstyle
\max_j \pair{f,\V{L_j} f} ~\ge~ \Omega\left(\frac{\log k}{k} \cdot
\frac{\Phi(G)}{\log n} \cdot \var(f)\right)\ ,\]
where $n = |V(G)|$.
\end{corollary}
In general, the above corollary is incomparable to Theorem
\ref{thm:kkl}, which gives a bound of $\Omega(\frac{\log k}{k} \cdot
\alpha \cdot \var (f))$.  We get a quantitatively better bound from
Theorem \ref{thm:kkl} if $\alpha \gg \frac{\Phi}{\log n}$.

An example where this is true is the following: Consider the
$R$-dimensional hypercube $H_R$ and consider $H_R^{\cart k}$. Though
this graph is isomorphic to the $kR$-dimensional hypercube, our notion
of \emph{influence} now translates to the number of edges cut along
one of the hypercubes $H_R$ i.e. the number of edges along one of the
$k$ blocks of $R$ coordinates each.

Assume we have a boolean function $f$ with variance $\Omega(1)$ on
$H_R^{\cart k}$.  Applying Theorem \ref{thm:kkl} to this instance, we
conclude that there must be a block of $R$ coordinates, along which
the fraction of edges cut is $\Omega\left(\frac{\log
    k}{kR}\right)$. Whereas, the bound that Corollary
\ref{cor:kkl-qary} gives is $\Omega\left(\frac{\log k}{k R^2}\right)$.

\paragraph{An example with maximum influence $o\inparen{\Phi \cdot
    \frac{\log k}{k}}.$}
Both Theorem \ref{thm:kkl} and Corollary \ref{cor:kkl-qary} generalize
the KKL theorem for the $k$-dimensional hypercube, and the ``tribes''
function which is known to be tight for the KKL theorem on the
hypercube also shows the tightness of the above theorems.
For the $k$-dimensional hypercube, the bound given by both theorems is
in fact $\Omega\inparen{\frac{\log k}{k} \cdot \Phi}$ since $\Phi =
\Omega(1)$ for the underlying graph (an
edge). 

However, both the above results give bounds which, in general, can be
much smaller than $\frac{\log k}{k} \cdot \Phi$, when applied to
$G^{\cart k}$ for an arbitrary graph $G$, for a function $f$ with
$\var(f) = \Omega(1)$. Below, we construct a family of examples to show
that this is necessary, and the theorems cannot be improved to give a
bound of $\Omega\inparen{\frac{\log k}{k} \cdot \Phi}$.

%

Consider the $R$-dimensional hypercube in which we identify any two
vertices which are the same after a cyclic permutation of the
coordinates. Formally, the vertex set is $\inparen{\B^R \setminus
  \inbraces{0^R, 1^R}} / {\cal C}$, where $\cal C$ is the group of
cyclic permutations on $\{1,\ldots, R\}$. Each vertex of the new graph
is then an equivalence class. There is an edge between two classes
$C_1$ and $C_2$ if there exist $u \in C_1$ and $v \in C_2$ such that
$(u,v)$ is an edge in the hypercube. We take this to be our graph $G$.

It follows from the KKL theorem for the hypercube that for the above
graph, $\Phi = \Omega\left(\frac{\log R}{R}\right)$. It was shown by
Devanur et al.  \cite{DevanurKSV06} (also see \cite{KollaL11} for a
more general proof) that for the above graph, the value of the
sparsest cut SDP (with triangle inequalities) is
$O\left(\nfrac{1}{R}\right)$. Since the SDP value is known to be an
upper bound on $\lambda$, we have $\lambda =
O\left(\nfrac{1}{R}\right)$ and hence $\alpha =
O\left(\nfrac{1}{R}\right)$ (in fact $\alpha =
\Theta\left(\nfrac{1}{R}\right)$ since collapsing vertices into
equivalence classes can only increase $\alpha$, which was
$\nfrac{2}{R}$ for the $R$-dimensional hypercube).

The bound given by Theorem \ref{thm:kkl} for maximum influence of a
coordinate in $\Gk$, when $\var (f) = \Omega(1)$, is then
$\Omega\inparen{\frac{1}{R} \cdot \frac{\log k}{k}} =
o\inparen{\Phi\cdot \frac{\log k}{k}}$ (if $R = \omega(1)$).  We now
show that this is tight.
\begin{claim}
Given $k \in \N$ and the graph $G$ as above with $R = k^{O(1)}$, there
exists a function $f: V(\Gk) \to \B$ such that $\var(f) = \Omega(1)$
and the influence of every coordinate is $O\inparen{\frac{1}{R} \cdot
\frac{\log k}{k}}$.
\end{claim}
\begin{proof}
A vertex $x \in V(\Gk)$ is of the form $(C_1, \ldots, C_k)$ where
$C_1, \ldots, C_k$ are equivalence classes in $\B^R$ as described
above. We take $t = \log_2(kR)$ and define $f$ as
\[ f(C_1, \ldots, C_k) = 1 ~~~\text{iff}~~~ \exists i \in [k]
~\text{and}~ u \in C_i ~\text{such that}~ u~\text{has $1^{t-1}0$ as a
substring} . \]
We first show that $\var(f) = \Omega(1)$ by bounding $\Pr[f=0]$. The
function $f$ is 0 on a vertex $(C_1,\ldots,C_k)$, when for each $i \in
[k]$, no $u \in C_i$ contains $1^{t-1}0$ as a substring. Since
different blocks are independent, we can simply estimate the
probability that within a single block, \emph{at least one} of the $u
\in C$ contains $1^{t-1}0$ (for $C$ chosen at random from
$V(G)$). This can be expressed as union of $R$ events, corresponding
to the position in $u$ where the 0 from $1^{t-1}0$ appears. Since no
two occurrences of the substring can overlap, two such events happen
with probability at most $1/2^{2t}$. Denoting by $p_1$ the above
probability for a single block, we get by inclusion-exclusion that
\[ \textstyle R \cdot \frac{1}{2^t} ~\geq~ p_1 ~\geq~ R \cdot \frac{1}{2^t} - R^2
\cdot \frac{1}{2^{2t}} .\]
Using the value of $t$, this gives $\frac{1}{k} \geq p_1 \geq
\inparen{1-\frac{1}{k}} \cdot \frac{1}{k}$. Since $\Pr[f=0]$ is simply
$(1-p_1)^k$, we get that
\[ \textstyle
\inparen{1-\frac{1}{k}}^k ~\leq~ \Pr[f=0] ~\leq~
\inparen{1-\frac{1}{k} + \frac{1}{k^2}}^k .\]
The above gives that $\var(f) = \Omega(1)$. Also, the definition of
$f$ is symmetric in all coordinates $i \in [k]$ and hence all
influences are equal. It remains to compute the influence of a
coordinate.

We estimate the probability that a random edge along the first
coordinate (say) is crossed by the cut that $f$ gives. Let $f(C_1,
C_2, \ldots,C_k) = 1$ and $f(C_1', C_2, \ldots,C_k) = 0$ where
$(C_1,C_1')$ is an edge in $G$. Then $C_1$ must have exactly one
substring of the form $1^{t-1} 0$, which happens with probability
$O(1/k)$ since the bounds for $p_1$ above also hold when $p_1$ is
taken to be the probability of \emph{exactly one occurrence} of the
substring.

Also, $C_1'$ must differ from $C_1$ in one of these $\log(kR)$
positions. For a fixed $C_1',$ this happens with probability
$\log(kR)/R$. Hence the fraction of edges crossed by the cut is
$O(\log(kR)/kR)$. Choosing $R = k^{O(1)}$ shows that this is $O((1/R)
\cdot (\log k/k)) = O(\alpha \cdot (\log k/k))$.
\end{proof}





%% file: irregular.tex
\section{Generalization to Reversible Markov Chains and Irregular
  Graphs}
\label{sec:irregular}
\paragraph{Preliminaries for reversible Markov chains.}
We first recall a few definitions for reversible Markov chains.  A
Markov chain $G$ on a finite state space $V(G)$ is defined by its
Kernel $\V{K}$, which is an $|V(G)| \times |V(G)|$ matrix satisfying,
\[\textstyle
\V{K}_{x,y} \ge 0, \ \sum_{y \in V(G)} \V{K}_{x,y} =1.\]
$\V{K}$ operates on the space of functions $f : V(G) \to \R$ as
$(\V{K}f)(x) = \sum_{y \in V(G)} \V{K}_{x,y} f(y).$ Let $\pi$ denote an
invariant measure for $\V{K},$ \emph{i.e.}, it satisfies,
\[\textstyle \sum_{x \in V(G)} \pi(x)\V{K}_{x,y} = \pi(y).\] 
Such a measure always exists, and is unique under a mild
irreducibility condition. We say that $(\V{K},\pi)$ is reversible if
it satisfies the \emph{detailed balance} condition for all $x,y \in
V(G),$ \emph{i.e.},
\[\forall\ x,y \in V(G),\ \pi(x)\V{K}_{x,y} = \pi(y)\V{K}_{y,x} .\]

\paragraph{General graphs to reversible Markov chains.}
Suppose we instead started with a general undirected graph $G$ with
its combinatorial adjacency matrix $\V{A},$ \emph{i.e.}, the weight of
the edge $(x,y) \in V(G) \times V(G)$ is given by $\V{A}_{x,y}.$ The
only constraint on $\V{A}$ is that it is a symmetric matrix with
non-negative entries. We define the degree of a vertex $x \in V(G)$ as
$d_x \defeq \sum_{y \in V(G)} \V{A}_{x,y}.$ Let $\V{D}$ be the
diagonal matrix with diagonal entries $\{d_x\}_{x\in V}.$ The random
walk on $G$ is a Markov chain with kernel $\V{K} \defeq
\V{D}^{-1}\V{A},$ with a reversible and stationary distribution,
$\pi(x) \defeq \frac{d_x}{\sum_{y} d_y}.$ For general graphs, we
define all the notions using this reversible Markov chain associated
with the graph. From now on, we'll work only with reversible Markov
chains.

The main theorems of this paper, Theorems~\ref{thm:kkl} and
\ref{thm:friedgut}, also hold for reversible Markov chains, or
equivalently, for irregular, weighted graphs (self-loops are also
permitted). We now give the required definitions in these cases.

Using these definitions, all the previous proofs go through without
modifications.

\subsection{Required definitions in case of reversible Markov chains.}
\begin{definition}[Cartesian product]
Given $k$ Markov chains $G_1,\ldots,G_k$, with state spaces $V(G_1)$,
$\ldots,V(G_k)$ and transition kernels $\V{K_1},\ldots,\V{K_k}$
respectively, their Cartesian product $G_1\cart \ldots \cart G_k$ is a
Markov chain on the state space $V(G_1) \times \ldots \times V(G_k),$
and its transition kernel is specified as follows:

Starting at state $(x_1,\ldots,x_k),$ pick $i \in [k]$ uniformly at
random. Pick a transition $x_i \rightsquigarrow y_i$ according to
$\V{K_i},$ and let the next state be
$(x_1,\ldots,x_{i-1},y_i,x_{i+1},\ldots,x_k).$
Equivalently, the transition kernel is $\frac{1}{k} \cdot \sum_{i=1}^k
\V{I} \otimes \ldots \otimes \V{K_i} \otimes \ldots \otimes \V{I},$
where $\V{K_i}$ is in the $i^\textrm{th}$ position.

Define $\Gk$ to be $G \cart \ldots \cart G$ ($k$ times).
\end{definition}
For the rest of the section, let $G$ be a reversible Markov chain with
state space $V(G),$ transition kernel $\V{K},$ and stationary
distribution $\pi(\cdot).$ We will assume that $G$ is irreducible, and
hence the stationary distribution $\pi$ is unique, as otherwise the
log-Sobolev constant $\alpha$ is 0 and our results become trivial.

The following simple observation tells us that the stationary
distribution on $\Gk$ is the product distribution.
\begin{claim}
Let $\pi(\cdot)$ be the unique stationary distribution for an
irreducible, reversible Markov chain $G$ defined on the state space $V(G).$
The unique stationary distribution for $\Gk$ is given by $\pi^{\otimes
  k},$ defined as $\pi^{\otimes k}(x_1,\ldots,x_k) \defeq \prod_{i=1}^k \pi(x_i).$ 
\end{claim}





The definition of inner product and norms for the space of functions
$V(G) \to \R$ remain the same except they use the stationary
distribution $\pi(\cdot)$ instead of the uniform distribution.
\[\pair{f,g} \defeq \av_{x\sim \pi} [f(x)g(x)] \ , \quad \norm{f}_2^2 \defeq
\pair{f,f} = \av_{x \sim \pi} [f(x)^2], \quad \text{and} \quad
\norm{f}_1 \defeq \Ex{x \sim \pi}{\abs{f(x)}} .\] 
Expectation and variance are also defined according to $\pi,$
\[\textstyle\var(f) \defeq \av_{x \sim \pi}f(x)^2 - \left(\av_{x\sim
    \pi}f(x)\right)^2.\]

For the space of functions $V(\Gk) \to \R,$ all the above notions are
defined using $\pi^{\otimes k},$ which is the stationary distribution
over $V(\Gk),$ however, as before we will use the same notation and
the corresponding space of functions will be clear from the
context. Importantly, under these definitions, we still have,
\begin{align*}
\textstyle\pair{f_1 \otimes \ldots \otimes f_k,g_1\otimes \ldots \otimes g_k}
= \prod_{i=1}^k \pair{f_i,g_i}.
\end{align*}
The normalized Laplacian $\V{L_G}$ is now defined as, $\V{L_G} \defeq
\V{I} - \V{K}.$ It is easy to verify that $\V{L_G}$ satisfies,
\[
\pair{f,\V{L_G}f} = \frac{1}{2}\sum_{x,y \in V(G)}
(f(x)-f(y))^2 \cdot \pi(x)\V{K}_{x,y} .\]
As before, we define the operator $\V{L_j}$ to be
\[\V{L_j} \defeq \V{I} \otimes \ldots \otimes \V{I} \otimes \V{L_G}
\otimes \V{I}\otimes \ldots \otimes \V{I},\] where the matrix
$\V{L_G}$ is in the $j^\textrm{th}$ position. Thus,
\begin{align*}
\av_j \V{L_j} & = \av_j ( \underbrace{\V{I} \otimes \ldots
\otimes \V{I}}_{j-1 \textrm{ copies }} \otimes \V{L_G} \otimes
\V{I}\otimes \ldots \otimes \V{I} ) = \av_j (
\underbrace{\V{I} \otimes \ldots \otimes \V{I}}_{j-1 \textrm{ copies
}} \otimes (\V{I} - \V{K}) \otimes \V{I} \otimes \ldots \otimes \V{I}
) \\ & = \V{I}\otimes \ldots \otimes \V{I} - \av_j (
\underbrace{\V{I} \otimes \ldots \otimes \V{I}}_{j-1 \textrm{ copies
}} \otimes \V{K} \otimes \V{I} \otimes \ldots \otimes \V{I})
\quad \stackrel{\textrm{by def.}}{=} \V{L_{\Gk}}\ .
\end{align*}
The set of ``edges'' along coordinate $j$ is $E_j(\Gk) \defeq \{\
(x,y)\ |\ \forall i \neq j,\ x_i = y_i\ \}.$ Thus, \[\pair{f,\V{L_j}f}
= \frac{1}{2}\sum_{(x,y) \in E_j(\Gk)} (f(x)-f(y))^2 \pi(x)
\V{K}_{x_j,y_j},\] where we note that for all $(x,y) \in E_j(\Gk),$
$\pi(x)\V{K}_{x_j,y_j} = \pi(y)\V{K}_{y_j,x_j}.$

\emph{Influence} of $f$ along coordinate $j$ is defined as
$\pair{f,\V{L_j}f},$ as before. \emph{Variance along $j^\textrm{th}$
  coordinate} is defined as before, except that each coordinate is
distributed independently according to $\pi.$
\[\textstyle \varj{j}(f) \defeq \Ex{x \setminus
  \{x_j\} \sim \pi^{\otimes (k-1)}}{\Ex{x_j \sim \pi}{f(x)^2} -
  \left(\Ex{x_j \sim \pi}{f(x)} \right)^2}.\]

Letting $\mathbbm{1}$ denote the all 1's vector, and $\Pi$ denote the
vector with the entries $\{\pi(x)\}_{x \in V(G)},$ define the operator
$\V{K_j}$ as
\[\V{K_j} \defeq \V{I}\otimes \V{I}\otimes\ldots \otimes
\left(\V{I}-\mathbbm{1}\Pi^{\top}\right) \otimes \ldots \otimes\V{I},\]
where the matrix $\V{I}-\mathbbm{1}\Pi^{\top}$ is in the
$j^{\textrm{th}}$ position. As before, for any $f : V(\Gk) \to \R$ and
any $j \in \{1,\ldots,k\},$ we have $\varj{j}(f) = \pair{f, \V{K_j}
  f}$.

\paragraph{Conductance.}
The volume of a set $\vol(S)$ is now defined as the measure of $S$
under $\pi$, $\vol(S) \defeq \sum_{v \in S} \pi(v).$ Conductance is now
defined as 
\[\Phi(G) \defeq \min_{\small \substack{S \subset V(G) \\ S \neq
\emptyset, V(G)}} \frac{\sum_{x\in S,y\in \bar{S}}
\pi(x)\V{K}_{x,y}}{2\vol(S)\vol(\bar{S})} = \min_{\small \substack{f
\from V(G) \to \{-1,1\} \\ \var(f) \neq 0}} \frac{\pair{f,\V{L_G}
f}}{2\var(f)}\]

\paragraph{Eigenfunctions and Eigenvalues.}
Since $G$ is a reversible Markov chain with the stationary
distribution $\pi(\cdot),$ for every $x,y \in V(G),$ we have
$\V{K}_{x,y} \pi(x) = \V{K}_{y,x} \pi(y).$ This implies that for
$f,g:V(G) \to \R,$ under the dot product defined according to $\pi,$
$\pair{\V{K}f, g} = \pair{f,\V{K}g}.$ Thus $\V{K}$ is a self-adjoint
operator, and so is $\V{L_G} = \V{I}-\V{K}.$ Thus, as before, it has
real eigenvalues $\lambda_0 = 0 \le \lambda_1 \le \ldots \le
\lambda_{n-1},$ and an orthonormal basis of eigenfunctions $v_0 =
\mathbbm{1},v_1,\ldots,v_{n-1}.$ All the properties of the
eigenfunctions of $\Gk$ described in Proposition~\ref{prop:Gk-eigen}
follow.

\paragraph{Log-Sobolev Constant.} The entropy of a function is now
defined with the expectations taken under $\pi.$
\begin{align*}
\Ent(f^2) &~\defeq~ \av_{x \sim \pi} [f(x)^2\log f(x)^2] - (\av_{x
\sim \pi} [f(x)^2])\log \av_{x \sim \pi} [f(x)^2]
\end{align*}
As before, the log-Sobolev constant of $G$ is the largest constant
$\alpha$ such that for all functions $f : V(G) \to \R$,
\[\pair{f,\V{L_G}f} ~\ge~ \frac{\alpha}{2} \cdot \Ent(f^2).\] Again, with
these definitions, if the log-Sobolev constant for $G$ is $\alpha,$
then the log-Sobolev constant for $\Gk$ is $\alpha/k$~\cite{DiaconisS96}.
Moreover, the isoperimetric constants defined above also satisfy the
inequalities $\alpha ~\le~ \lambda_1 ~\le~ 2\Phi.$


%% file: applications.tex
\section{Applications to integrality gaps for Sparsest Cut}
\label{appendix:applications}
In this section, we show that the Cartesian product is the right
method of \emph{padding} Sparsest Cut integrality gap instances. We
recall that the Sparsest Cut value of a graph $G$ is determined by the
conductance of $G$ up to a factor of $2$. In this section, we will
work with conductance.

Firstly, we show that, given a graph $G$ that has a conductance value
of $\Phi$, $\Gk$ has a conductance value of $\frac{1}{k}\Phi$. We also
show a similar statement for the optimum value for several families of
SDP relaxations of Sparsest Cut. Informally, we show that, given a
graph $G,$ with an SDP value of $\opt$, $\Gk$ has an SDP value of
$\frac{1}{k}\opt$, where the SDP is any one of several common classes
of SDP relaxations for conductance. In particular, the ratio of the
two values is preserved.

We first prove the theorem about conductance. This theorem is similar
to (special cases of) the ones proved in \cite{HoudreT96} and
\cite{ChungT98}. We include a proof for completeness. The following
proof approach is also included in (the journal version of)
\cite{HoudreT96}.
\begin{theorem}[Conductance Value]
  Given a graph $G$, we can relate the Conductance of $\Gk$ and $G$ as
  follows,
  \[\Phi(\Gk) ~=~ \frac{1}{k}\Phi(G)\ .\]
\end{theorem}
\begin{proof}
Let us first prove the simple direction $\Phi(\Gk) \le
\frac{1}{k}\Phi(G)$. In order to prove this, fix a function $f \from
V(G) \to \{-1,1\}$ that achieves
$\frac{\pair{f,\V{L_G}f}}{2\var(f)} = \Phi(G).$

Now, define $g \from V(\Gk) \to \R$ as $g(v_1,\ldots,v_k) =
f(v_1)$. It is easy to see that $\var(g) = \var(f)$. Let us
compute $\pair{g,\V{L_{\Gk}}g}$.
\begin{align*}
\pair{g,\V{L_{\Gk}}g} & = \frac{1}{2} \av_{(x,y) \in E(\Gk)} (g(x)-g(y))^2 \\
& = \frac{1}{2} \av_{(x,y) \in E(\Gk)} (f(x_1)-f(y_1))^2 \\
& = \frac{1}{2k} \av_{x_2,\ldots,x_k} \av_{(x_1,y_1) \in E(G)} (f(x_1)-f(y_1))^2  
\shortintertext{(for all other edges the contribution is 0)}
& = \frac{1}{2k} \av_{x_2,\ldots,x_k} 2\pair{f,\V{L_G}f} \\
& = \frac{2\Phi(G)}{k} \var(f) = \frac{2\Phi(G)}{k} \var(g)
\end{align*}
Thus, $\Phi(\Gk) \le \frac{1}{k}\Phi(G)$.

Now, let us prove that $\Phi(\Gk) \ge \frac{1}{k} \Phi(G)$. Fix a set
$S \subseteq V(\Gk)$. Let $f$ denote the $\{-1,1\}$-valued indicator
function for $S$. From Lemma~\ref{lem:Lj-conductance}, we know that
$\pair{f,\V{L_j}f} \ge 2 \Phi(G) \cdot \varj{j}(f).$ Averaging over
all $j \in [k],$
\[\textstyle
 \pair{f,\V{L_{\Gk}}f} = \av_{j}{\pair{f,\V{L_j}f}} \ge 2
\frac{\Phi(G)}{k} \cdot \sum_{j} \varj{j}(f) \ge 2
\frac{\Phi(G)}{k} \cdot \var(f),\] 
where the last inequality follows from
Lemma~\ref{lem:vj-var-comparison}, which gives $\sum_{j \in [k]}
\varj{j}(f) \ge \var(f)$. Hence $\Phi(\Gk) \ge
\frac{1}{k}\Phi(G)$.
\end{proof}

Now, we will show that if we consider an SDP relaxation for Sparsest
Cut, the SDP value for $\Gk$ will be $\frac{1}{k}\opt$, where $\opt$
is the SDP value for $G$. This theorem holds for the following
families of SDP relaxations - standard SDP, SDP with triangle
inequalities, SDP with $k$-gonal inequalities, standard SDP with $t$
levels of Sherali-Adams constraints (for any fixed $t$) and SDP at
$t^\textrm{th}$ levels of Lasserre hierarchy.
\begin{theorem}[SDP Value for Sparsest Cut]
Let $\Psi$ be an SDP relaxation for Sparsest Cut that is one of the
following: the standard relaxation with $t$ levels of SA variables,
the standard relaxation lifted to $t$ level of Lasserre hierarchy ($t$
is arbitrary). Then, denoting the relaxation $\Psi$ applied to $G$ as
$\Psi(G)$ and the optimum for $\Psi(G)$ by $\opt(G)$,
\[\opt(\Gk) \le \frac{1}{k}\opt(G)\]
\end{theorem}
\begin{proof}
Let $\{v_u\}_{u \in V(G)}$ be an optimum solution to the SDP
$\Psi(G)$. We will construct a solution for $\Psi(\Gk)$ with objective
value $\frac{1}{k}\opt(G)$.

\paragraph{Standard SDP.}
The Standard SDP for Sparsest Cut applied to graph $G$ is the following:
\begin{align}
\min & \qquad\av_{\{x,y\} \in E(G)}\norm{v_x - v_y}_2^2, \label{eq:obj}\\
\textrm{s.t.} &\qquad \av_{x,y} \norm{v_x - v_y}_2^2 = 1.\label{eq:spread}
\end{align}

Consistent with our notation, we define the $\ell_2$ norm of two
vectors $v_x$ and $v_y$ as $\norm{v_x - v_y}_2^2 \defeq \av_i
\left(v_x^{(i)} - v_y^{(i)}\right)^2.$ Define the vectors $v_x =\bigoplus_i
v_{x_i}$ where $x = (x_1,\ldots,x_k) \in V(\Gk)$ and $\bigoplus$
denotes the direct-sum operation. First, consider the objective
function (\ref{eq:obj}),
\begin{align*}
\av_{(x,y) \in E(\Gk)} \|v_x - v_y\|^2 & = \av_j \av_{(x,y) \in E_j(\Gk)} \|v_x - v_y\|^2 \\
& = \frac{1}{k} \av_j \av_{(x,y) \in E_j(\Gk)} \|v_{x_j} - v_{y_j}\|^2 \\
& = \frac{1}{k} \av_j \opt(G) = \frac{1}{k} \opt(G).
\end{align*}

Now, let us verify that the \emph{spreadness} constraint
(\ref{eq:spread}) is satisfied.
\begin{align*}
  \av_{x,y \in V(\Gk)} \|v_x - v_y\|^2 = \av_{x,y \in V(\Gk)} \av_j
  \|v_{x_j}-v_{y_j}\|^2 = \av_{x,y \in V(G)} \|v_x - v_y\|^2.
\end{align*}

\paragraph{$k$-gonal inequalities.}
The relaxation $\Psi$ may contain a certain family of constraints
that must be satisfied by the vectors (e.g. triangle inequalities,
$k$-gonal inequalities). The best approximation algorithm by Arora,
Rao and Vazirani~\cite{ARV09} uses triangle inequalities in the SDP
relaxation.
\[\forall\ x,y,z \in V(G),\qquad \norm{v_x-v_y}_2^2 + \norm{v_y-v_z}_2^2
\ge \norm{v_x-v_z}_2^2\]

We know that the vectors we have constructed satisfy $\pair{v_x,v_y} =
\av_j \pair{v_{x_j},v_{y_j}}$. It follows that as long as these set of
constraints is invariant under a permutation of the vertices and is
linear in the dot product of the vectors (which is true in case of the
examples mentioned), the new vectors also satisfy the corresponding
constraints by linearity.

\paragraph{SDP with Sherali-Adams Constraints.} Let us now consider
the Sherali-Adams constraints. The Sherali-Adams SDP for level $t$
requires that for every set $T \subseteq V(G)$ of at most $t$
vertices, there must be a distribution $\calD_T$ on the integral
assignments to $T$ ($-1,1$ valued). There are consistency constraints
requiring that for two sets $T_1,T_2$, the marginals of the
distributions $\calD_{T_1},\calD_{T_2}$ on $T_1 \cap T_2$ are
identical. There are also consistency constraints with the vector
solution requiring that
\[\forall x,y \in V(G),\ \pair{v_x,v_y} = \av_{(z_x,z_y) \sim \calD_{\{x,y\}}} z_x z_y.\]
Now, let us define the distributions that correspond to the solution
for $\Psi(\Gk)$. In order to sample from $\calD_{\{x^1,\ldots,x^t\}}$,
we pick a random $j \in [k]$, draw a sample $z \sim
\calD_{\{x_j^1,\ldots,x_j^t\}}$ and output $z$ (Note that we are
concerned only about sets, so if an element appears more than once,
it's treated as if it appeared just once). Consistency of marginals
follows easily because of linearity. Again, since $\pair{v_x,v_y} =
\av_j \pair{v_{x_j},v_{y_j}}$, by linearity, we get that the
distributions constructed are consistent with the vectors.

\paragraph{Lasserre SDP.}
The SDP at $t^{\textrm{th}}$ round of the Lasserre Hierarchy has
vectors for every subset $S$ with at most $t$ vertices in the graph and has the
following family of constraints,
\[\pair{v_{S_1},v_{S_2}} = \pair{v_{T_1},v_{T_2}} \text{ whenever }
S_1 \Delta S_2 = T_1 \Delta T_2\ .\]
Given vectors for $S \subseteq V(G)$, we wish to construct vectors for
$T \subseteq V(\Gk)$ so that the consistency constraints are
satisfied. Define sets $T_j \subseteq V(G)$ as follows.
\[T_j \defeq \left\{\ y\ |\ \#\{x \in T|x_j = y\} \textrm{ is odd } \right\}\]
Now, define $v_{T}$ as the direct sum of these $k$ vectors, $v_{T}
\defeq \oplus_{j=1}^k v_{T_j}$. It is easy to verify that these
vectors satisfy the consistency constraints.
\end{proof}